\documentclass[a4paper,12pt,reqno]{amsart}
          \usepackage{amssymb}
	  \usepackage{amsmath}
          \usepackage{amsfonts}
          \usepackage[english]{babel}
          \usepackage[utf8]{inputenc}

\usepackage[margin=2.5cm]{geometry}

 \usepackage[unicode,colorlinks,plainpages=false,hyperindex=true,bookmarksnumbered=true,bookmarksopen=false,pdfpagelabels]{hyperref}
 \hypersetup{urlcolor=cyan,linkcolor=blue,citecolor=red,colorlinks=true}
\makeatletter
\renewcommand*{\p@section}{\,}
\renewcommand*{\p@subsection}{\S\,}
\renewcommand*{\p@subsubsection}{\S\,}
\makeatother

\newcommand{\be}{\begin{equation}}
\newcommand{\ee}{\end{equation}}
\newcommand{\bea}{\begin{eqnarray}}
\newcommand{\eea}{\end{eqnarray}}

\newcommand{\ba}{\begin{aligned}}
\newcommand{\ea}{\end{aligned}}

\newtheorem{thm}{Theorem}[section]
\newtheorem{cor}[thm]{Corollary}
\newtheorem{lem}[thm]{Lemma}
\newtheorem{prop}[thm]{Proposition}
\theoremstyle{definition}
\newtheorem{rem}[thm]{Remark}

 \numberwithin{equation}{section}

\newcommand{\CC}{\ensuremath{\mathbb{C}}}
\newcommand{\N}{\ensuremath{\mathbb{N}}}
\newcommand{\R}{\ensuremath{\mathbb{R}}}
\newcommand{\tr}{\operatorname{tr}}
\newcommand{\Mat}{\operatorname{Mat}}
\newcommand{\sgn}{\operatorname{sgn}}
\newcommand{\ic}{\ensuremath{\mathfrak{i}}}
 \def\cP{{\mathcal P}}

\def\1{{\boldsymbol 1}}
\def\U{{\rm U}}

\def\hA{\hat{A}}
\def\hB{\hat{B}}
\def\C{\mathbb{C}}
\def\Gn{{\rm GL}(n,\C)}
\def\Gd{{\rm GL}(d,\C)}
\def\Gl{{\rm GL}(\ell,\C)}
\def\gn{\mathfrak{gl}(n,\C)}

\def\gl{{\mathfrak{gl}}(\ell,\C)}
\def\cD{{\mathcal D}}
\def\cF{{\mathcal F}}
\def\cH{{\mathcal H}}
\def\dz {\left.\frac{d}{dz}\right|_{z=0}}
\def\cL{{\mathcal L}}
\def\cA{{\mathcal A}}
\def\cB{{\mathcal B}}
\def\cG{{\mathcal G}}
 \newcommand\br[1]{\{ #1 \}}

\begin{document}

\author[M.~Fairon and L.~Feh\'er]{} 

\title[Poisson structures on $\Mat_{n\times d}(\C) \times \Mat_{d\times n}(\C)$]{A decoupling property of some Poisson
structures on $\Mat_{n\times d}(\C) \times \Mat_{d\times n}(\C)$
supporting  $\Gn \times \Gd$ Poisson--Lie symmetry} 


\maketitle

\medskip
\begin{center}
M.~Fairon${}^{a}$ and L.~Feh\'er${}^{b,c}$
\\

\bigskip
${}^a$School of Mathematics and Statistics, University of Glasgow\\
University Place, G12 8QQ Glasgow, United-Kingdom\\
e-mail: maxime.fairon@glasgow.ac.uk

\medskip
${}^b$Department of Theoretical Physics, University of Szeged\\
Tisza Lajos krt 84-86, H-6720 Szeged, Hungary\\
e-mail: lfeher@physx.u-szeged.hu

\medskip
${}^c$Institute for Particle and Nuclear Physics\\
Wigner Research Centre for Physics\\
 H-1525 Budapest, P.O.B.~49, Hungary

\end{center}

\medskip
\begin{abstract}
We study a holomorphic Poisson structure defined on the linear space
$S(n,d):= \Mat_{n\times d}(\C) \times \Mat_{d\times n}(\C)$
that is covariant under the natural left actions of the standard $\Gn$ and $\Gd$  Poisson--Lie groups.
The Poisson brackets of the matrix elements contain quadratic and constant terms, and the Poisson tensor is
non-degenerate on a dense subset.
Taking the $d=1$ special case gives a Poisson structure on $S(n,1)$, and we construct
a local Poisson map from the Cartesian product of $d$ independent copies of $S(n,1)$ into $S(n,d)$,
which is a holomorphic diffeomorphism in a neighborhood of zero.
The Poisson structure on $S(n,d)$ is the complexification of
a real Poisson structure on $\Mat_{n\times d}(\C)$ constructed by the authors and Marshall,
where a similar decoupling into $d$ independent copies was observed.
We also relate our construction to a Poisson structure on $S(n,d)$ defined by
Arutyunov and Olivucci in the treatment of the
complex trigonometric spin Ruijsenaars--Schneider system
by Hamiltonian reduction.
\end{abstract}

{\linespread{0.5}\tableofcontents}

\section{Introduction}
\label{sec:I}
In this paper we prove a remarkable `decoupling property' of a holomorphic Poisson structure
defined on the space
\be
S(n,d):= \Mat_{n\times d}(\C) \times \Mat_{d\times n}(\C),
\label{I1}\ee
which appeared in recent derivations of trigonometric spin Ruijsenaars--Schneider models \cite{KZ} by Hamiltonian reduction
\cite{AO,FFM}.
The decoupling means that the Poisson algebra of $S(n,d)$   will  be realized using $d$ independent
(pairwise Poisson commuting)
copies of the Poisson algebra of $S(n,1)$.
The spaces $S(n,d)$ are defined for arbitrary pairs
of natural numbers, but the decoupling requires that both $n$ and $d$ are greater than $1$.
Our result is expected to be useful, for example,  for the further studies of the holomorphic
spin Ruijsenaars--Schneider systems.

To set the stage, for any natural number $\ell$ we introduce the Drinfeld--Jimbo classical $r$-matrix $r^\ell$ by
\be
r^\ell := \frac{1}{2} \sum_{1\leq j<k\leq \ell} E_{jk}(\ell) \wedge E_{kj}(\ell),
\label{I2}\ee
where $E_{jk}(\ell)$ is the usual elementary matrix of size $\ell \times \ell$.
We also need
\be
r^\ell_{\pm}:= r^\ell \pm \frac{1}{2} I^\ell
\quad
\hbox{with}\quad I^\ell:= \sum_{j,k=1}^\ell E_{jk}(\ell) \otimes E_{kj}(\ell).
\label{I3}\ee
 Note that for $\ell =1$   $r^\ell = 0$ and
$I^\ell$ can be viewed as $1\otimes 1$.
Denoting the elements of $S(n,d)$ as pairs $(A, B)$, and employing the standard
tensorial notation \cite{A,FT},
 the pertinent Poisson bracket can be written as follows:
\bea
  &&\br{A_1,A_2}_\kappa=-\kappa (r^n A_1 A_2 + A_1 A_2 r^d)\,,\nonumber \\
  &&\br{B_1,B_2}_\kappa=-\kappa (B_1 B_2 r^n + r^d B_1 B_2) \,,   \label{I4}\\
  && \br{A_1,B_2}_\kappa= \kappa (B_2 r_+^n A_1 + A_1 r_{+}^dB_2 +  {C}_{12}^{n\times d})\,. \nonumber
\eea
 Here, we use the notations \eqref{I2}, \eqref{I3} together with
 \be
 C_{12}^{n\times d}:=\sum_{i=1}^n \sum_{\alpha=1}^d  E_{i\alpha}^{n\times d}\otimes E_{\alpha i}^{d\times n},
\label{C12} \ee
 where $E_{i\alpha}^{n\times d}\in \Mat_{n\times d}(\C)$ is the elementary matrix having a single non-zero
 entry, equal to $1$, at the $i\alpha$ position.
 One could fix the arbitrary constant $\kappa\in \C^*$ without loss of generality, but it will be advantageous not to do so.

The Poisson structure \eqref{I4} represents the complexification
of a $\U(n)\times \U(d)$ covariant
real Poisson structure
on $\Mat_{n\times d}(\C) \simeq \R^{2nd}$ considered in \cite{FFM}.
By simple changes of variables (see below)
it also reproduces the holomorphic Poisson bracket
defined on $S(n,d)$ by Arutyunov and Olivucci \cite{AO}.
In the papers mentioned it was natural to assume that $n >1$, but here we assume only
that either $n$ or $d$ is greater than $1$.
The $d=1$ (or $n=1$) cases provide the building blocks from which the general $S(n,d)$ cases
will be realized via the decoupling.

The above Poisson brackets have remarkable Poisson--Lie covariance properties.
 (For background on the theory of Poisson--Lie groups, one may consult, for example,
\cite{A,KS,STS1,STS2}.)
To describe these, we equip the group $\Gl$ with the standard multiplicative Poisson bracket
given in tensorial notation by
\be
\{ g_1, g_2\}_G^\kappa := \kappa [g_1 g_2, r^\ell].
\label{I6}\ee
The subscript $G$
expresses that this Poisson bracket lives on the group $G=\Gl$.

Then the linear left-action of $\Gn$ on $S(n,d)$, defined by
\be
\Gn \times S(n,d)   \ni (g, A,B) \mapsto (g A, B g^{-1}) \in S(n,d)
\label{I7}\ee
enjoys the Poisson--Lie property, which means that the map \eqref{I7} is Poisson if $\Gn$
is equipped with the bracket \eqref{I6} for $\ell =n$ and $S(n,d)$ is equipped with
the bracket \eqref{I4}.
Similarly, the linear left-action of $\Gd$, given by
\be
\Gd \times S(n,d)   \ni (g, A,B) \mapsto (Ag^{-1}, g B) \in S(n,d),
\label{I8}\ee
also has the Poisson--Lie property, where $\Gd$ is equipped with  the bracket \eqref{I6}, for $\ell =d$.

Now we describe our main result, which was motivated by an analogous result of \cite{FFM}.
Let us introduce the group
\be
D(\ell):= \Gl \times \Gl.
\label{Dl}\ee
This is the Drinfeld double of the Poisson--Lie group $\Gl$. The dual Poisson--Lie group,
$\Gl^*$, is the subgroup of $D(\ell)$ consisting of pairs $(h_+, h_-)$, where $h_+$ and $h_-$ are invertible
upper triangular and, respectively, lower triangular matrices whose respective diagonal entries are inverses
of each other, i.e., $(h_-)_{jj} = 1/ (h_+)_{jj}$ for $j=1,\dots, \ell$.
 Consider the space $S(n,1)$, with elements $(a^1,b^1)$, endowed with the Poisson bracket
\eqref{I4}. Introduce the (locally defined) map
\be
(g_+, g_-): S(n,1) \to \Gn^*
\label{gpm}\ee
by the following definition:
\begin{equation} \label{Eq:gplus}
 (g_+)_{jj}=\sqrt{G_j /G_{j+1}}\,, \quad (g_+)_{jk}=\frac{a_j^1b_k^1}{\sqrt{G_k G_{k+1}}}\,\text{ for}\,\, j<k\,,
\end{equation}
and
\begin{equation} \label{Eq:gminus}
 (g_-^{-1})_{jj}=\sqrt{G_j /G_{j+1}}\,, \quad (g_-^{-1})_{jk}=\frac{a_j^1 b_k^1}{\sqrt{G_j G_{j+1}}}\,\text{ for}\,\, j>k\,,
\end{equation}
using the functions
\begin{equation} \label{Eq:Gj}
 G_j=1 +\sum_{k=j}^n a_k^1b_k^1\,, \quad G_0=G_{n+1}=1\,,
\end{equation}
which are well-defined only locally, including a neighborhood of zero.

\begin{thm} \label{Thm:Main}
For any $n$ and $d$ greater than $1$, take $d$ copies of $S(n,1)$, each equipped with the Poisson bracket \eqref{I4},
 and denote their elements by $(a^\alpha, b^\alpha)$,
$\alpha =1,\dots,d$.
Let $(a,b)$ stand for the collection of the $(a^\alpha, b^\alpha)$,
$A^\alpha$ and $B^\alpha$ for the columns and the rows of the  matrices $(A,B) \in S(n,d)$, respectively.
Define
the (local) map
\begin{equation}\label{m}
m: S(n,1) \times \cdots \times S(n,1) \to S(n,d)
\end{equation}
by the formulae
$A^1(a,b) = a^1$, $B^1(a,b)= b^1$ and, for $\alpha \geq 2$,
\begin{subequations}
\begin{align}
&A^\alpha (a,b)= g_{+}(a^1,b^1)  \cdots g_{+}(a^{\alpha-1}, b^{\alpha -1}) a^\alpha, \label{Eq:MainA}\\
& B^\alpha(a,b) = b^\alpha g^{-1}_-(a^{\alpha-1}, b^{\alpha -1})\cdots    g^{-1}_-(a^1,b^1). \label{Eq:MainB}
\end{align}
\end{subequations}
Then the map $m$  is a  local, holomorphic Poisson diffeomorphism from the $d$-fold product
Poisson space $S(n,1) \times \cdots \times S(n,1)$
to $S(n,d)$, where $S(n,1)$ and $S(n,d)$ are equipped with the relevant Poisson brackets of the form \eqref{I4}.
\end{thm}

The fundamental property of the map $(g_+, g_-)$ \eqref{gpm} is
the factorization identity
\be
\1_n + a^1 b^1 = g_+(a^1,b^1) g_-(a^1, b^1)^{-1}.
\label{factid1}\ee
Introducing
\be
\cG_\pm(a,b):= g_\pm(a^1,b^1) \cdots g_\pm(a^d, b^d),
\label{cGpmdef}\ee
the identity \eqref{factid1}  and formulae of Theorem \ref{Thm:Main} imply the further identity
\be
\1_n + A(a,b) B(a,b) = \cG_+(a,b) \cG_-(a,b)^{-1}.
\label{factid2}\ee
These properties, which are easily verified, actually motivated our construction.
Their meaning will be enlightened in Section \ref{Sec:Cov} (see Remark \ref{rem:factid}) utilizing the theory of Poisson--Lie moment maps.

\medskip
We can also give an analogous realization of the Poisson bracket \eqref{I4} on $S(n,d)$ in terms of $n$ copies
of the Poisson bracket on $S(1,d)$.
Such a map can be obtained by combining  Theorem \ref{Thm:Main} with the swap map $\nu$ from $S(n,d)$ to $S(d,n)$
that operates according to
\be
\nu: (A,B) \mapsto (\eta^d B \eta^n, \eta^n A \eta^d),
\ee
where for any $\ell\in \N$ we let $\eta^\ell:=\sum_{i=1}^\ell E_{i, \ell + 1 - i}(\ell)$.
It is easily seen that
\be
\nu: (S(n,d),\{\ ,\ \}_\kappa) \to ( S(d,n), \{\ ,\ \}_{-\kappa})
\ee
is a Poisson diffeomorphism.

In addition, we shall present decoupling results for the `oscillator Poisson brackets'
of Arutyunov--Olivucci \cite{AO}, who introduced two Poisson structures
on $S(n,d)$. Denoting the elements of $S(n,d)$ now as pairs $(\cA, \cB)$, one of their
Poisson structures, called $\{\ ,\ \}_\kappa^+$, is given by
\bea
  &&\br{\cA_1,\cA_2}_{\kappa}^+= \kappa (r^n \cA_1 \cA_2 - \cA_1 \cA_2 r^d)\,,\nonumber \\
  &&\br{\cB_1,\cB_2}_{\kappa}^+=\kappa (\cB_1 \cB_2 r^n -r^d \cB_1 \cB_2) \,, \label{+PB}\\
  &&\br{\cA_1,\cB_2}_{\kappa}^+=\kappa (-\cB_2 r^n_+ \cA_1 + \cA_1 r^d_{-}\cB_2)-C_{12}^{n\times d}\,. \nonumber
\eea
Their other Poisson bracket, called $\{\ ,\ \}_\kappa^-$, is obtained from this one by replacing
$(\cA, \cB)$ by $(\cA \eta^d, \eta^d \cB)$ in the above formula.
In fact, we have two different decoupling result for the Poisson bracket $\{\ ,\ \}_\kappa^+$.
The first one is obtained by combining Theorem \ref{Thm:Main} with the following simple lemma.

\begin{lem} \label{lem:lemI2}
Let $\xi_A$ and $\xi_B$ be arbitrary constants for which $\xi_A \xi_B = -\frac{1}{\kappa}$.
Then the map
\be
\xi: (A, B) \mapsto (\cA, \cB):= ( \xi_A A \eta^d, \xi_B \eta^d B)
\ee
is a Poisson diffeomorphism from
$(S(n,d), \{\ ,\ \}_\kappa)$ to
$(S(n,d), \{\ ,\ \}_{-\kappa}^+)$.
\end{lem}

An alternative decoupling map from $(S(n,1), \{\ ,\ \}_\kappa)^{\times d}$  to $(S(n,d), \{\ ,\ \}_{\kappa}^+)$
will be presented in Section \ref{Sec:X}.

\begin{rem}
It is known that the brackets $\{\ ,\ \}_\kappa$ and $\{\ ,\ \}_\kappa^+$ satisfy
the Jacobi identity, but the interested reader can also check this by
routine calculation.
\end{rem}

\section{Basic facts about Poisson--Lie groups} \label{Sec:PL}

We will recall the embedding of the Poisson--Lie group $\Gl$ and its dual into their Drinfeld double $D(\ell)$.
Then we will present the  notion of the Poisson--Lie moment map.
We do not give proofs here, since the relevant statements can be found in many  reviews \cite{A,KS,STS1,STS2}.

Let us consider the complex Lie group $D(\ell)$ \eqref{Dl}
and equip its Lie algebra,
\be
\cD(\ell) := \gl \oplus \gl,
\label{cDl}\ee
with the non-degenerate, invariant bilinear form
\be
\langle (U,V), (X,Y) \rangle_\kappa:= \frac{1}{\kappa} \left(\tr(UX) - \tr (VY)\right),
\label{pairing}\ee
using a constant $\kappa \in \C^*$.
Let us also introduce the triangular decomposition
\be
\gl = \gl_> + \gl_0 + \gl_<,
\ee
where $\gl_0$ is the set of diagonal matrices, while
$\gl_>$ (resp. $\gl_<$) contains the upper (resp. lower) triangular matrices with zero diagonal.
Then $\cD(\ell)$ can be represented as the vector space direct sum of the isotropic subalgebras
\be
\gl_\delta := \{ (X,X) \mid X\in \gl\}
\label{gldelta}\ee
and
\be
\gl_\delta^*:= \{ (Y_> + Y_0, Y_< - Y_0)\mid Y_>\in \gl_>,\,  Y_< \in \gl_<,\, Y_0\in \gl_0\}.
\ee
In \eqref{gldelta} the subscript $\delta$ indicates that $\gl_\delta$ is the diagonal embedding of
$\gl$ into $\cD(\ell)$ \eqref{cDl}.
We may identify $\gl$ with the diagonal subalgebra $\gl_\delta$, and  identify
its linear dual space with the subalgebra $\gl^*_\delta$.
We also let $\Gl_\delta$ and $\Gl_\delta^*$ denote the subgroups of $D(\ell)$ corresponding to the
subalgebras in the decomposition
\be
\cD(\ell) = \gl_\delta + \gl_\delta^*.
\ee

The group $D(\ell)$ carries a natural multiplicative Poisson structure.
To describe it, let us take arbitrary bases $T^a$  of $\gl_\delta$ and $T_a$ of $\gl_\delta^*$ that are
in duality with respect to the pairing \eqref{pairing}.  The Poisson bracket of two holomorphic functions $\cF$ and $\cH$ on $D(\ell)$
is given by
\be
\{\cF, \cH\}_{D}^\kappa :=\sum_{a=1}^{n^2} \left(   (\nabla_{T^a} \cF) (\nabla_{T_a} \cH) -(\nabla'_{T^a} \cF) (\nabla'_{T_a} \cH) \right),
\ee
where for any $T\in \cD(\ell)$ we have
\be
(\nabla_T \cF)(p) =  \dz \cF(e^{z T} p),
\quad
(\nabla_T' \cF)(p) =  \dz \cF(p e^{z T}),
\quad \forall p \in D(\ell).
\label{PBdoub}\ee
It is well-known that  $\Gl_\delta$ and $\Gl_\delta^*$ are Poisson submanifolds of $D(\ell)$, and we equip
them with the inherited Poisson structures.

The above Poisson structures can be conveniently presented in terms of the functions
given by the matrix elements on the respective groups.
Denoting the elements of $D(\ell)$ as pairs $(u,v)$, and employing the tensorial notation of the Faddeev school, one has
\be
\{ u_1, u_2\}_{D}^\kappa = \kappa [ u_1 u_2, r^\ell],
\,\,
\{ v_1, v_2\}_{D}^\kappa = \kappa [ v_1 v_2, r^\ell],
\,\,
\{ u_1, v_2\}_{D}^\kappa = \kappa [ u_1 v_2, r_+^\ell],
\label{DlPB}\ee
using the $r$-matrices \eqref{I2} and \eqref{I3}.
On the subgroup $\Gl_\delta$ with elements denoted $(g,g)$, this reduces to the bracket \eqref{I6}.
The group $\Gl_\delta^*$ consists of the pairs $(h_+, h_-)\in D(\ell)$ for which $h_+$ (resp. $h_-$) is upper triangular (resp. lower triangular)
and the diagonal part of $h_+$ is the inverse of the diagonal part of $h_-$.
Restriction from $D(\ell)$ gives the following Poisson bracket on this dual group:
\be
\{ h_{\pm,1}, h_{\pm,2}\}_*^\kappa = \kappa [ h_{\pm,1} h_{\pm,2}, r^\ell],\,\,
\{ h_{+,1}, h_{-,2}\}_*^\kappa = \kappa [ h_{+,1} h_{-,2}, r_+^\ell].
\label{starPB}\ee
We stress that $D(\ell)$, $\Gl\equiv \Gl_\delta$ and $\Gl^* := \Gl_\delta^*$  with the above Poisson brackets are Poisson--Lie groups.
This means, for example, that the group product $D(\ell) \times D(\ell) \to D(\ell)$ is a Poisson map.

Let us briefly explain  how \eqref{DlPB} follows from \eqref{PBdoub}. For $T=(X,Y)\in \cD(\ell)$, the derivatives of the matrix elements are
\be
\nabla_T u_{ij} = (Xu)_{ij}, \quad \nabla_T' u_{ij} = (u X)_{ij},\quad
\nabla_T v_{ij} = (Yv)_{ij}, \quad \nabla_T' v_{ij} = (v Y)_{ij} .
\ee
For any dual bases $T^a = (X^a, X^a)$ and $T_a = (Z_a, W_a)$, one can calculate that
\be
X^a \otimes Z_a = - \kappa r_-^\ell \quad\hbox{and}\quad X^a \otimes W_a  = - \kappa r_+^\ell.
\label{rid}\ee
By using these relations, one readily obtains \eqref{DlPB} from \eqref{PBdoub}.

There is an important mapping of $\Gl^*$ onto $\Gl$, which is given by
\be \label{chi}
\chi: (h_+, h_-) \mapsto h:=h_+ h_-^{-1}.
\ee
This mapping is $2^n$ to $1$, since the image does not change if we replace $(h_+, h_-)$ by $(h_+\tau, h_- \tau)$ for any diagonal matrix $\tau$ whose entries
are taken from the set $\{ +1, -1\}$. The  map $\chi$ yields a holomorphic
diffeomorphism\footnote{\label{Foot1}The multi-valued inverse of the map $\chi$ is closely related to the Gauss decomposition \cite[\S6.1]{NS} of regular matrices.
If we Gauss decompose any $g\in \Gl$ as $g = g_> g_0 g_<$,  and have $g_0 = h_0^2$ for the diagonal constituent, then
$g= h_+ h_-^{-1}$ with $h_+ =g_> h_0$ and $h_-^{-1} = h_0 g_<$, so that $(h_+, h_-)\in \Gl^*$.  The subtlety is that we have to take
the square root of the diagonal matrix $g_0$.}
between respective neighborhoods of the identity elements.
Moreover, it is a Poisson map with respect to the so-called Semenov-Tian-Shansky Poisson structure \cite{A,STS1}  on $\Gl$:
\be
\{ h_1, h_2\}_{\mathrm{STS}}^\kappa =  \kappa \left(  h_1 r_-^\ell h_2 + h_2 r_+^\ell h_1 - h_1 h_2 r^\ell - r^\ell h_1 h_2 \right).
\label{STSPB}\ee
With this Poisson bracket, $\Gl$ can serve, at least locally, as a model of the dual Poisson--Lie group $\Gl^*$.
We note in passing that this quadratic Poisson bracket naturally extends to a Poisson structure on $\gl$, which is
compatible with its linear Lie--Poisson  bracket.

We now recall \cite{Lu} what is meant by a moment map for a Poisson action of $\Gl$. Suppose that $\Gl$ acts on a holomorphic
Poisson manifold $(\cP, \{\ ,\ \}_\cP)$ in such a way that the action map, $\Gl\times \cP \to \cP$,
is Poisson, where the product Poisson structure on $\Gl \times \cP$ is built from the bracket
$\{\ ,\ \}_G^\kappa$  \eqref{I6} on $\Gl$ and $\{\ ,\ \}_\cP$ on $\cP$.
For any $X\in \gl$, let $X_\cP$ be the vector field on $\cP$ given by the flow of $\exp(t X)$.
We can take the derivative $\cL_{X_\cP} \cF$ of any holomorphic function on $\cP$.
We then say that a holomorphic map $(\phi_+, \phi_-): \cP \to \Gl_\delta^*$ is the (Poisson--Lie) moment map for the action
if it satisfies the following two conditions.
First, we must have the equality
\be
\cL_{X_\cP} \cF = \langle (X,X), \{ \cF, (\phi_+, \phi_-)\}_\cP (\phi_+, \phi_-)^{-1} \rangle_\kappa
\label{momcond}\ee
for all $X$ and $\cF$. Second, we also require that $(\phi_+, \phi_-)$ is  a Poisson map with respect
to the bracket \eqref{starPB} on the dual group.
This second condition is equivalent to the requirement that the map
\be
\phi:= \phi_+ \phi_-^{-1}: \cP \to \Gl
\ee
is Poisson with respect to the Semenov-Tian-Shansky bracket \eqref{STSPB} on $\Gl$.
Indeed, the Semenov-Tian-Shansky bracket is just the push-forward of the Poisson bracket \eqref{starPB} on the
dual group $\Gl^*$.
The first condition can also be recast in terms of the map $\phi$, as we shall see in our concrete example in the next section.

\section{Covariance properties of the Poisson structure (\ref{I4})} \label{Sec:Cov}

We now characterize the behavior of
the Poisson bracket \eqref{I4} on $S(n,d)$ under the natural left-action of $\Gn$.
Throughout this section, $d\geq 1$ and $n\geq 2$, otherwise they are arbitrary.
The statements presented below can also be obtained
as consequences of known \cite{AO,FFM}  analogous properties of the Arutyunov--Olivucci Poisson bracket \eqref{+PB}.
We sketch the  proofs in order to make this paper basically self-contained.

\begin{prop} \label{pr:Cov1}
Consider the natural action of $\Gn$ on $S(n,d)$,
defined by
\be
g\cdot (A,B):= (gA, B g^{-1}), \quad \forall g\in \Gn,\,\, (A,B)\in S(n,d).
\label{H1}\ee
With respect to the Poisson structures \eqref{I4} and \eqref{I6},
this is a Poisson action.
\end{prop}

\begin{proof}
This is very easy and goes as follows.
We can calculate $\{ g_1 A_1, g_2 A_2\}$ using the product Poisson structure defined by
combining \eqref{I4} and \eqref{I6}. This gives
\bea
\{ g_1 A_1, g_2 A_2\} &=& \{ g_1, g_2\}^\kappa_G A_1 A_2 + g_1 g_2 \{ A_1, A_2\}_\kappa \nonumber\\
\qquad &=& -\kappa (r^n g_1A_1 g_2 A_2 + g_1A_1 g_2A_2 r^d),
\eea
which agrees with the Poisson bracket on $S(n,d)$.
The next line of \eqref{I4} is handled in the same way.
Finally, one needs to show that
\be
\{ g_1 A_1, B_2 g_2^{-1}\} =
\kappa (B_2g_2^{-1} r_+^n g_1A_1 + g_1A_1 r_{+}^dB_2g_2^{-1} +  {C}_{12}^{n\times d}).
\ee
We refrain from spelling this out, but note that
the direct verification of this  equality relies on the identity
$g_1 C_{12}^{n\times d} = C_{12}^{n\times d} g_2$.
\end{proof}

\begin{prop} \label{pr:Cov2}
 Suppose that $(\phi_+, \phi_-): S(n,d) \to \Gn^*$ is a (possibly only locally defined) moment map
for the Poisson action \eqref{H1}. Then the condition \eqref{momcond} is equivalent to the equalities
\be
 \{ A_1, \phi_{\pm,2}\}_\kappa = - \kappa r_\mp^n A_1 \phi_{\pm,2}
 \quad\hbox{and}\quad
\{ B_1, \phi_{\pm,2}\}_\kappa = \kappa B_1 r_\mp^n  \phi_{\pm,2},
\label{mom1}\ee
where the usual tensorial notation is employed.
For $\phi:= \phi_+ \phi_-^{-1}$, these relations imply
\be
\{ A_1, \phi_2\}_\kappa = \kappa (\phi_2 r_+^n - r_-^n \phi_2) A_1
\quad\hbox{and}\quad
\{ B_1, \phi_2\}_\kappa = \kappa B_1 (r_-^n \phi_2 - \phi_2 r_+^n).
\label{mom2}\ee
\end{prop}
\begin{proof}
Let $T^a = (X^a, X^a)$ and $T_a= (Z_a, W_a)$ be dual bases of $\gn_\delta$ and $\gn_\delta^*$.
Consider an arbitrary matrix element $A_{i\alpha}$ as a function on $S(n,d)$. Its
 derivative along the vector field
induced by $X^a\in \gn$  equals $(X^aA)_{i\alpha}$, and \eqref{momcond} gives the identity
\be
(X^a A)_{i\alpha} = \langle T^a, \{ A_{i\alpha}, (\phi_+, \phi_-) \}_\kappa (\phi_+^{-1}, \phi_-^{-1}) \rangle_\kappa.
\ee
Since $T^a$ is a basis of $\gn_\delta$, this implies that
\be
\{ A_{i\alpha}, (\phi_+, \phi_-)\}_\kappa (\phi_+^{-1}, \phi_-^{-1}) = (X^a A)_{i\alpha} T_a,
\ee
This is equivalent to the relations
\be
\{ A_{i\alpha}, \phi_+\}_\kappa = (X^a A)_{i\alpha} Z_a \phi_+
\quad\hbox{and}\quad
 \{ A_{i\alpha}, \phi_-\}_\kappa = (X^a A)_{i\alpha} W_a \phi_-\,.
 \ee
 By using the identities \eqref{rid}, these two equations are just the componentwise form of
 the first tensorial formulae in \eqref{mom1}. The relations involving $B$ are verified in the same way.
 The equalities in \eqref{mom1} are converted into those in \eqref{mom2} by a short calculation.
Since the matrix elements of $A$ and $B$ form a coordinate system on $S(n,d)$, the
proof is complete.
\end{proof}

\begin{prop} \label{pr:Cov3}
Define the map $\Gamma: S(n,d) \to \gn$ by the formula
\be
\Gamma(A,B) = \1_n + AB.
\ee
As a consequence of the Poisson brackets \eqref{I4}, this map satisfies the
relation
\be
\{ \Gamma_1, \Gamma_2\}_\kappa =  \kappa \left(  \Gamma_1 r_-^n \Gamma_2 + \Gamma_2 r_+^n \Gamma_1 - \Gamma_1 \Gamma_2 r^n - r^n \Gamma_1 \Gamma_2 \right)
\label{Ga1}\ee
together with
\be
\{ A_1, \Gamma_2\}_\kappa = \kappa (\Gamma_2 r_+^n - r_-^n \Gamma_2) A_1
\quad\hbox{and}\quad
\{ B_1, \Gamma_2\}_\kappa = \kappa B_1 (r_-^n \Gamma_2 - \Gamma_2 r_+^n).
\label{Ga2}\ee
In a neighborhood of zero, $\Gamma$ can be represented in the form $\Gamma = \Gamma_+ \Gamma_-^{-1}$ so that
$(\Gamma_+, \Gamma_-): S(n,d) \to \Gn^*$  serves  (locally) as the Poisson--Lie moment map for the action \eqref{H1}.
\end{prop}

\begin{proof}
The equalities \eqref{Ga1} and \eqref{Ga2} can be verified by an easy  calculation.
For this, one needs to use the identities $A_1 A_2 I^d = I^n A_1 A_2$ and
$I^n A_1 = A_2 C_{12}^{n\times d}$.
Since $\Gamma(0)=\1_n$, it is clear that $\Gamma$ admits a unique factorization of the form $\Gamma = \Gamma_+ \Gamma_-^{-1}$ if we restrict
$(A,B)$ to be near enough to zero, require the continuity of $\Gamma_\pm$ and impose the condition
$\Gamma_\pm(0)= \1_n$. The so obtained $(\Gamma_+, \Gamma_-)$ can be written as holomorphic functions
of the matrix entries of $\Gamma$.
Equation \eqref{Ga1} entails that $(\Gamma_+, \Gamma_-)$ gives a Poisson map into the dual group $\Gl^*$ carrying the brackets \eqref{starPB}, and the relations
\eqref{Ga2} are equivalent to the  moment map conditions given in \eqref{mom1}.
Here, we used the coincidence of \eqref{mom2} with \eqref{Ga2} and that $\Gamma$ and $(\Gamma_+, \Gamma_-)$
are related by a local Poisson diffeomorphism.
\end{proof}

\begin{rem}
A natural generalization of Proposition \ref{pr:Cov3} holds around  an arbitrary point $(A_0, B_0)\in S(n,d)$ for which
$(\1_n + A_0 B_0)$ is an invertible matrix. See footnote \ref{Foot1} for the construction of $(\Gamma_+, \Gamma_-)$.
\end{rem}

\begin{rem} \label{rem:factid}
We observe from  Proposition \ref{pr:Cov3} and the first factorization identity \eqref{factid1} that $(g_+, g_-)$ given by  (\ref{gpm})
is nothing but
the (local) $\Gn^*$-valued  Poisson--Lie moment map on $S(n,1)$.
For $d>1$,  the meaning of the second factorization identity \eqref{factid2} is that the (local) moment map $(\Gamma_+, \Gamma_-)$ mentioned
in Proposition  \ref{pr:Cov3}  satisfies the equality
\be
(\Gamma_\pm \circ m)(a,b) = \cG_\pm (a, b),
\ee
where $m$ is the map of Theorem \ref{Thm:Main} and $\cG_\pm$ are defined in \eqref{cGpmdef}.
\end{rem}

\section{Derivation of Theorem \ref{Thm:Main}} \label{Sec:Proof}

The Poisson structure is derived in \ref{ss:ProofPB}, and we prove that $m$ is a diffeomorphism in \ref{ss:ProofDiff}.
The strategy of the derivation is analogous to \cite[Lemma 5.1]{FFM}.

\subsection{The Poisson structure} \label{ss:ProofPB}

\subsubsection{Preparation and notations}
The product Poisson structure on $S(n,1)^{\times d}$ can be written in tensor notation using the pairs $(a^\alpha,b^\alpha)$  as follows
\begin{subequations}
 \begin{align}
  \br{a_1^\alpha,a_2^\beta}_\kappa=&-\kappa \delta_{\alpha \beta} \,r^n\, a_1^\alpha a_2^\beta \,, \qquad
  \br{b_1^\alpha,b_2^\beta}_\kappa=-\kappa\delta_{\alpha\beta} b_1^\alpha b_2^\beta \,r^n\,, \label{Eq:PBdaa} \\
  \br{a_1^\alpha, b_2^\beta}_\kappa=&\kappa \delta_{\alpha\beta} \left(b_2^\beta \,r_+^n\, a_1^\alpha + \frac12 a_1^\alpha b_2^\beta +C_{12}^{n\times 1}\right) \,.\label{Eq:PBdbb}
 \end{align}
\end{subequations}
For fixed $\alpha\in \{1,\ldots,d\}$, we use the pair $(a^\alpha,b^\alpha)$ to define  $G^\alpha_j$ locally by \eqref{Eq:Gj}, then the upper and lower triangular matrices $g_{+,\alpha},g_{-,\alpha}^{-1}$ by \eqref{Eq:gplus}--\eqref{Eq:gminus}.
Using \eqref{factid1} in the form
\begin{equation} \label{Eq:Fact}
 g_{+,\alpha}g_{-,\alpha}^{-1}=\1_n+a^\alpha b^\alpha\,,
\end{equation}
we can combine Propositions \ref{pr:Cov2} and \ref{pr:Cov3} for each copy $S(n,1)$, and we obtain
\begin{subequations}
 \begin{align}
\br{a_1^\alpha,(g_{+,\beta})_2}_\kappa=&-\kappa \delta_{\alpha\beta} \,r_-^n\, a_1^\alpha (g_{+,\beta})_2\,, \quad
\br{a_1^\alpha,(g_{-,\beta}^{-1})_{2}}_\kappa=\kappa \delta_{\alpha\beta}(g_{-,\beta}^{-1})_{2} \,r_+^n\, a_1^\alpha\,,  \label{Eq:agplus} \\
\br{b_1^\alpha,(g_{+,\beta})_2}_\kappa=&\kappa \delta_{\alpha\beta} b_1^\alpha \,r_-^n\,  (g_{+,\beta})_2\,, \quad
\br{b_1^\alpha,(g_{-,\beta}^{-1})_{2}}_\kappa=-\kappa \delta_{\alpha\beta}b_1^\alpha (g_{-,\beta}^{-1})_{2} \,r_+^n\,. \label{Eq:bgplus}
 \end{align}
\end{subequations}
We also use the formulae \eqref{starPB} to write
\begin{subequations}
 \begin{align}
\br{(g_{+,\alpha})_1,(g_{+,\beta})_2}_\kappa=&\kappa \delta_{\alpha\beta} [(g_{+,\alpha})_1 (g_{+,\beta})_2, r^n]  \,, \label{Eq:ggplus} \\
\br{(g_{-,\alpha}^{-1})_1,(g_{-,\beta}^{-1})_2}_\kappa=&-\kappa \delta_{\alpha\beta}[(g_{-,\alpha}^{-1})_1 (g_{-,\beta}^{-1})_2, r^n]\,,\label{Eq:ggminus}\\
    \br{(g_{+,\alpha})_{1},(g_{-,\beta}^{-1})_{2}}_\kappa=&
-\kappa \delta_{\alpha\beta} \left((g_{+,\alpha})_1 r_+^n (g_{-,\beta}^{-1})_2 - (g_{-,\beta}^{-1})_2 r_+^n (g_{+,\alpha})_1 \right)\,.  \label{Eq:ggpm}
 \end{align}
\end{subequations}
\begin{rem}
 The  involution  $ \iota :  S(n,1)^{\times d} \to S(n,1)^{\times d}$ defined by
 \begin{equation} \label{Eq:iotad}
\iota(a^1,b^1,\dots, a^d,b^d)=((b^1)^T,(a^1)^T,\dots,(b^d)^T,(a^d)^T )\,,
 \end{equation}
is an anti-Poisson automorphism by Remark \ref{Rem:inv} in the appendix.
Using this map, we observe the following identities of matrix-valued functions
\begin{equation} \label{Eq:iotagg}
 g_{+,\alpha}\circ \iota =(g_{-,\alpha}^{-1})^T \,, \quad g_{-,\alpha}^{-1}\circ \iota=g_{+,\alpha}^T\,.
\end{equation}
 The consistency of the anti-Poisson property with the Poisson brackets collected above is straightforward to check.
\end{rem}

To ease computations, we introduce
\begin{equation}
h_+^\alpha=g_{+,1}\cdots g_{+,\alpha}\,, \quad
h_-^\alpha=g_{-,\alpha}^{-1}\cdots g_{-,1}^{-1}\,,\quad 1\leq \alpha \leq d\,,
\end{equation}
with $h_+^0=\1_n=h_-^0$, so that \eqref{Eq:MainA}--\eqref{Eq:MainB} become
\begin{equation} \label{Eq:ABbis}
 A^\alpha=h_{+}^{\alpha-1}a^\alpha\,, \quad B^\alpha=b^\alpha h_-^{\alpha-1}\,, \quad 1\leq \alpha \leq d\,.
\end{equation}
It will also be convenient to introduce for $1\leq \alpha \leq \gamma \leq d$
\begin{equation}
 \begin{aligned}
  h_+^{\alpha;\gamma}=&g_{+,\alpha}\cdots g_{+,\gamma}\,, \quad h_+^{\alpha;\alpha}=g_{+,\alpha},\quad
  h_-^{\alpha;\gamma}=g_{-,\gamma}^{-1} \cdots g_{-,\alpha}^{-1}\,, \quad h_-^{\alpha;\alpha}=g_{-,\alpha}^{-1}\,,
 \end{aligned}
\end{equation}
and we set $h_+^{\alpha;\alpha-1}=\1_n=h_-^{\alpha;\alpha-1}$.
We note in particular that under the involution $\iota$ \eqref{Eq:iotad} which satisfies \eqref{Eq:iotagg}, we can write
\begin{equation} \label{Eq:iotahh}
h_+^{\alpha;\gamma}\circ \iota=(h_-^{\alpha;\gamma})^T\,, \quad
h_-^{\alpha;\gamma}\circ \iota=(h_+^{\alpha;\gamma})^T\,.
\end{equation}

\subsubsection{Preliminary lemmae}

\begin{lem} \label{L:bas1}
The following identities hold
 \begin{subequations}
  \begin{align}
   \br{a_1^\alpha,(h_+^\beta)_2}_\kappa=& -\kappa\delta_{(\alpha \leq \beta)} (h_+^{\alpha-1})_2 \,r_-^n\, a_1^\alpha (h_+^{\alpha;\beta})_2\,, \\
   \br{b_1^\alpha,(h_+^\beta)_2}_\kappa=& \kappa \delta_{(\alpha \leq \beta)}  (h_+^{\alpha-1})_2 b_1^\alpha \,r_-^n\, (h_+^{\alpha;\beta})_2\,, \\
   \br{a_1^\alpha,(h_-^\beta)_2}_\kappa=& \kappa \delta_{(\alpha \leq \beta)} (h_-^{\alpha;\beta})_2 \,r_+^n\, a_1^\alpha (h_-^{\alpha-1})_2  \,, \\
   \br{b_1^\alpha,(h_-^\beta)_2}_\kappa=& -\kappa \delta_{(\alpha \leq \beta)} b_1^\alpha (h_-^{\alpha;\beta})_2 \,r_+^n\, (h_-^{\alpha-1})_2  \,.
  \end{align}
 \end{subequations}
Here, the value of $\delta_{(\alpha \leq \beta)}$ is $1$ if the condition $\alpha \leq \beta$ holds, and is zero otherwise.
\end{lem}
\begin{proof}
We have from \eqref{Eq:agplus} that $\br{a_1^\alpha,(g_{+,\gamma})_2}_\kappa$ vanishes identically if $\gamma\neq \alpha$. By definition of $h_+^\beta$, we thus get
\begin{equation}
 \br{a_1^\alpha,(h_+^\beta)_2}_\kappa= (h_+^{\alpha-1})_{2}\br{a_1^\alpha,(g_{+,\alpha})_{2}}_\kappa (h_+^{\alpha+1;\beta})_{2}\,,
\end{equation}
if $\alpha \leq \beta$, while it vanishes for $\beta<\alpha$. We then get the first equality from \eqref{Eq:agplus}.
The second equality is found in the same way, and the following two are obtained by applying the anti-Poisson automorphism $\iota$.
\end{proof}

\begin{lem} \label{L:bas2}
The following identities hold
 \begin{subequations}
  \begin{align}
  \br{(h_+^\alpha)_1,(h_+^\beta)_2}_\kappa\stackrel{\alpha \leqslant \beta}{=}&
\kappa\left((h_+^\alpha)_1(h_+^\alpha)_2 \, r^n\, (h_+^{\alpha+1;\beta})_2  - r^n (h_+^\alpha)_1(h_+^\beta)_2\right) \,, \\
  \br{(h_+^\alpha)_1,(h_-^\beta)_2}_\kappa\stackrel{\alpha \leqslant \beta}{=}&
\kappa \left((h_-^\beta)_2\,r_+^n\,(h_+^\alpha)_1 - (h_+^\alpha)_1(h_-^{\alpha+1;\beta})_2\, r_+^n\, (h_-^\alpha)_2 \right)\,, \\
  \br{(h_+^\alpha)_1,(h_-^\beta)_2}_\kappa \stackrel{\alpha \geqslant \beta}{=}&
 \kappa \left((h_-^\beta)_2\,r_+^n\,(h_+^\alpha)_1 - (h_+^\beta)_1\, r_+^n\,(h_+^{\beta+1;\alpha})_1 (h_-^\beta)_2 \right)\,.
  \end{align}
 \end{subequations}
\end{lem}
\begin{proof}
For the first identity, since $\alpha \leq \beta$ we use the decomposition
 \begin{equation*}
\br{(h_+^\alpha)_1,(h_+^\beta)_2}_\kappa=\sum_{\gamma=1}^\alpha (h_+^{\gamma-1})_1(h_+^{\gamma-1})_2
\br{(g_{+,\gamma})_1,(g_{+,\gamma})_2}_\kappa (h_+^{\gamma+1;\alpha})_1 (h_+^{\gamma+1;\beta})_2 \,,
 \end{equation*}
and note that the Poisson bracket appearing in the sum is given by \eqref{Eq:ggplus}.
This directly leads to the claimed result.
For the second identity, we write for $\alpha \leq \beta$
  \begin{equation*}
\br{(h_+^\alpha)_1,(h_-^\beta)_2}_\kappa=\sum_{\gamma=1}^\alpha (h_+^{\gamma-1})_1 (h_-^{\gamma+1;\beta})_2
\br{(g_{+,\gamma})_1,(g_{-,\gamma}^{-1})_2}_\kappa (h_+^{\gamma+1;\alpha})_1 (h_-^{\gamma-1})_2 \,,
 \end{equation*}
then we use \eqref{Eq:ggpm} to get the desired result.  The case $\alpha \geq \beta$ is obtained in a similar way.
\end{proof}
Note that the identities from Lemmae \ref{L:bas1} and \ref{L:bas2} can be used with $h_{\pm}^0=\1_n$  as well.

\subsubsection{The Poisson brackets $\br{A_1,A_2}_\kappa$ and $\br{B_1,B_2}_\kappa$}

We note that obtaining $\br{A_1,A_2}_\kappa$ in \eqref{I4} is equivalent to deriving
\begin{equation}
 \label{Eq:AAexpl}
 \br{A_1^\alpha,A_2^\beta}_\kappa = -\kappa \left(r^n A_1^\alpha A_2^\beta + \frac12 \sgn(\alpha-\beta) A_1^\beta A_2^\alpha \right)\,.
\end{equation}
This follows by spelling out the action of $r^d$ using \eqref{I2}.
In order to get \eqref{Eq:AAexpl}, we note that \eqref{Eq:ABbis} yields
  \begin{equation}
  \begin{aligned}
\br{A_1^\alpha,A_2^\beta}_\kappa=&
 \br{(h_+^{\alpha-1})_1,(h_+^{\beta-1})_2}_\kappa a_1^\alpha a_2^\beta
+(h_+^{\beta-1})_2 \br{(h_+^{\alpha-1})_1, a_2^\beta}_\kappa a_1^\alpha   \\
& +(h_+^{\alpha-1})_1 \br{a_1^\alpha,(h_+^{\beta-1})_2}_\kappa  a_2^\beta
+(h_+^{\alpha-1})_1(h_+^{\beta-1})_2 \br{a_1^\alpha,a_2^\beta}_\kappa \,.
  \end{aligned}
 \end{equation}
 We can then use \eqref{Eq:PBdaa} and Lemmae \ref{L:bas1}, \ref{L:bas2} to reduce this expression. If $\alpha=\beta$, we directly get
\begin{equation}
 \br{A_1^\alpha,A_2^\alpha}_\kappa = -\kappa r^n A_1^\alpha A_2^\alpha \,.
\end{equation}
If $\alpha<\beta$, we get
\begin{equation}
 \br{A_1^\alpha,A_2^\beta}_\kappa = -\kappa r^n A_1^\alpha A_2^\alpha + \frac{\kappa}{2} (h_+^{\alpha-1})_1(h_+^{\alpha-1})_2\, I^n\, a_1^\alpha (h_+^{\alpha;\beta-1})_2 a_2^\beta \,.
\end{equation}
Upon using the identity
\begin{equation}
 I^n\, a_1^\alpha (h_+^{\alpha;\beta-1}a^\beta)_2  = a_2^\alpha  (h_+^{\alpha;\beta-1}a^\beta)_1\,,
\end{equation}
we find
\begin{equation}
 \br{A_1^\alpha,A_2^\beta}_\kappa = -\kappa \,r^n A_1^\alpha A_2^\beta + \frac{\kappa}{2}  A_1^\beta A_2^\alpha \,.
\end{equation}
Thus, we have derived \eqref{Eq:AAexpl} for all $\alpha \leq \beta$, hence it holds for all $\alpha,\beta$ by antisymmetry.
We can check that we obtain the claimed Poisson bracket for  $\br{B_1,B_2}_\kappa$ either by a direct computation, or using the anti-Poisson automorphism $\iota$ \eqref{Eq:iotad} under which $A\circ \iota=B^T$.

\subsubsection{The Poisson bracket $\br{A_1,B_2}_\kappa$}

We now use that obtaining $\br{A_1,B_2}_\kappa$ in \eqref{I4} is equivalent to deriving
\begin{equation}
 \label{Eq:ABexpl}
 \br{A_1^\alpha,B_2^\beta}_\kappa = \kappa \left(B_2^\beta \,r^n_+\, A_1^\alpha + \frac12 \delta_{\alpha\beta} A_1^\alpha B_2^\beta + \delta_{\alpha\beta} \sum_{\mu<\alpha}A_1^\mu B_2^\mu + \delta_{\alpha\beta} C_{12}^{n\times 1} \right)\,.
\end{equation}

We have by \eqref{Eq:ABbis} that
  \begin{equation}
  \begin{aligned} \label{Eq:pfAB1}
\br{A_1^\alpha,B_2^\beta}_\kappa=&
 b_{2}^\beta \br{(h_+^{\alpha-1})_{1},(h_-^{\beta-1})_{2}}_\kappa a_{1}^\alpha
+\br{(h_+^{\alpha-1})_{1}, b_{2}^\beta }_\kappa a_{1}^\alpha  (h_-^{\beta-1})_{2} \\
& +(h_+^{\alpha-1})_{1} b_{2}^\beta \br{a_{1}^\alpha,(h_-^{\beta-1})_{2}}_\kappa
+ (h_+^{\alpha-1})_{1} \br{a_{1}^\alpha,b_{2}^\beta}_\kappa (h_-^{\beta-1})_{2}   \,,
  \end{aligned}
 \end{equation}
which can be computed using \eqref{Eq:PBdbb} and Lemmae \ref{L:bas1}, \ref{L:bas2}.
If $\alpha<\beta$, only the first and third sums in \eqref{Eq:pfAB1} do not trivially vanish, and we find that
  \begin{equation}
\br{A_1^\alpha,B_2^\beta}_\kappa= \kappa B_2^\beta\,r_+^n\,A_1^\alpha\,. \label{Eq:pfAB2}
 \end{equation}
If $\alpha>\beta$, we also obtain \eqref{Eq:pfAB2} by a similar computation.
If $\alpha=\beta$,  only the first and fourth sums in \eqref{Eq:pfAB1} are nonzero, and we obtain
\begin{equation}
 \label{Eq:pfAB3}
 \br{A_1^\alpha,B_2^\alpha}_\kappa = \kappa \left(B_2^\alpha \,r^n_+\, A_1^\alpha + \frac12  A_1^\alpha B_2^\alpha
 + (h_+^{\alpha-1})_1 C_{12}^{n\times 1}  (h_-^{\alpha-1})_2 \right)\,.
\end{equation}
We deduce from \eqref{Eq:Fact} that $h_{+}^{\alpha-1}h_-^{\alpha-1}=\1_n+\sum_{\mu<\alpha} A^\mu B^\mu$, which implies
   \begin{equation}
(h_+^{\alpha-1})_1 C_{12}^{n\times 1}  (h_-^{\alpha-1})_2 = C_{12}^{n\times 1} + \sum_{\mu<\alpha} A_1^\mu B_2^\mu\,.
 \end{equation}
Thus, we can write $\br{A_1^\alpha,B_2^\beta}_\kappa$ for all $\alpha,\beta$ in the desired form \eqref{Eq:ABexpl}.

\subsection{Diffeomorphism property} \label{ss:ProofDiff}

Let us consider a point where the map $m$ \eqref{m} is well-defined, i.e. we can construct $g_{\pm,\alpha}=g_{\pm}(a^\alpha,b^\alpha)$ for $\alpha=1,\ldots,d$
using the formulae \eqref{Eq:gplus}, \eqref{Eq:gminus} with \eqref{Eq:Gj}.
In a sufficiently small neighborhood of this point, the entries of the matrices $g_{\pm,\alpha}$ are analytic functions, hence $m$ is holomorphic.
From the image of this neighborhood, we can define inductively
\begin{equation}
\begin{aligned}
 &a^1=A^1\,, \,\,\, b^1=B^1\,, \qquad
 a^2=g_{+,1}^{-1}A^2\,, \,\,\, b^2=B^2 g_{-,1}\,,\quad \ldots\,, \\
 &a^d=g_{+,d-1}^{-1}\cdots g_{+,1}^{-1}A^d\,, \,\,\, b^d=B^d g_{-,1} \cdots g_{-,d-1}\,,
\end{aligned}
\end{equation}
which is the inverse of the map $m$. The inverse map is holomorphic  since the elements $g_{\pm,\alpha}$ are analytic functions in $(A_i^\beta,B_i^\beta)$ for $\beta\leq\alpha$.

\section{A decoupling property of the Arutyunov--Olivucci bracket} \label{Sec:X}

We now derive an alternative  realization of the Poisson algebra \eqref{+PB},
as was promised after Lemma \ref{lem:lemI2}.
For this purpose, we take $d$ copies of $(S(n,1),\{\ ,\ \}_\kappa)$ with variables
$(a^\alpha, b^\alpha)$ for $\alpha =1,\dots, d$, and define the new variables $(\hA^\alpha, \hB^\alpha)$
as follows:
\begin{subequations}
\begin{align}
&\hA^\alpha (a,b)= g_{+}^{-1}(a^d,b^d)  \cdots g_{+}^{-1}(a^{\alpha}, b^{\alpha}) a^\alpha, \label{X1}\\
& \hB^\alpha(a,b) = b^\alpha g_-(a^{\alpha}, b^{\alpha})\cdots    g_-(a^d,b^d), \label{X2}
\end{align}
\end{subequations}
using the functions introduced previously in \eqref{Eq:gplus} and \eqref{Eq:gminus}.

\begin{thm} \label{thm:X1}
The map $F: (a,b) \mapsto (\hA, \hB)$ given by \eqref{X1}, \eqref{X2},
where $(\hA^\alpha, \hB^\alpha)$ denote the columns and the rows, respectively, of the
 matrices $(\hA,\hB)\in S(n,d)$, is a local Poisson diffeomorphism
\be
F: (S(n,1), \{\ ,\ \}_\kappa)^{\times d} \to (S(n,d), \{\ ,\ \}_{\kappa}'),
\ee
where $\{\ ,\ \}'_\kappa$ denotes the Poisson structure on $S(n,d)$ defined by
\bea
  &&\br{\hA_1,\hA_2}'_\kappa= \kappa (r^n \hA_1 \hA_2 - \hA_1 \hA_2 r^d)\,, \nonumber\\
  &&\br{\hB_1,\hB_2}'_\kappa= \kappa (\hB_1 \hB_2 r^n - r^d \hB_1 \hB_2) \,, \label{hatPB} \\
  &&\br{\hA_1,\hB_2}'_\kappa= \kappa(-\hB_2 r^n_+ \hA_1 + \hA_1 r^d_{-} \hB_2 + C_{12}^{n\times d})\,. \nonumber
\eea
\end{thm}
\begin{proof}
 The calculation of the Poisson brackets of the functions in \eqref{X1}, \eqref{X2} is in principle straightforward,
and follows the derivation of Theorem \ref{Thm:Main} made in \ref{ss:ProofPB}.

The fact that the map $F$ is a local
diffeomorphism is similar to the argument used in \ref{ss:ProofDiff}.  We begin by observing
the identities
\begin{equation} \label{Eq:hatotherL}
 \1_n-{g}_{+,\alpha+1}\cdots {g}_{+,d}\hat{A}^\alpha \hat{B}^\alpha {g}_{-,d}^{-1}\cdots {g}_{-,\alpha+1}^{-1} = {g}_{+,\alpha}^{-1} {g}_{-,\alpha}
 \quad\hbox{for}\quad
 \alpha = d,\dots, 1,
\end{equation}
which follow from \eqref{X1}, \eqref{X2} using $g_{\pm, \alpha} = g_\pm(a^\alpha, b^\alpha)$, with $g_{\pm, d+1}:= \1_n$, and applying
the analogue of \eqref{factid1} for all $\alpha$.
Then, picking $\hat A$  and $\hat B$ near zero, we define the functions $(\hat g_{+, \alpha}, \hat g_{-,\alpha}) \in \Gn^*$ for $1\leq \alpha \leq d$,
by considering the factorization problems
\begin{equation} \label{Eq:hatd}
 \1_n-\hat{A}^d \hat{B}^d = \hat{g}_{+,d}^{-1} \hat{g}_{-,d}
\end{equation}
and iteratively
\begin{equation}\label{Eq:hatother}
\1_n- \hat{g}_{+,\alpha+1}\cdots \hat{g}_{+,d}\hat{A}^\alpha \hat{B}^\alpha \hat{g}_{-,d}^{-1}\cdots \hat{g}_{-,\alpha+1}^{-1} = \hat{g}_{+,\alpha}^{-1} \hat{g}_{-,\alpha}
 \quad\hbox{for}\quad
 \alpha = d-1,\dots, 1.
\end{equation}
This procedure uniquely specifies $\hat{g}_{\pm,\alpha}$ for all $\alpha$ if we set $\hat{g}_{\pm,\alpha} =\1_n$ for vanishing $\hat A$ and $\hat B$, and further require that these matrices
depend continuously on $\hat A, \hat B$  in an open neighborhood of zero.
As the final step, we define
\begin{equation} \label{Eq:hatinv}
 a^\alpha=\hat{g}_{+,\alpha}\hat{g}_{+,\alpha+1}\cdots \hat{g}_{+,d}\hat{A}^\alpha\,, \quad
 b^\alpha=\hat{B}^\alpha \hat{g}_{-,d}^{-1}\cdots \hat{g}_{-,\alpha+1}^{-1} \hat{g}_{-,\alpha}^{-1}\,.
\end{equation}
The definitions guarantee that if on the left-hand sides of \eqref{Eq:hatd}  and \eqref{Eq:hatother} we use \eqref{X1} and \eqref{X2}, then we obtain
\begin{equation}
 \hat{g}_{+,\alpha}=g_+(a^\alpha,b^\alpha)\,, \quad \hat{g}_{-,\alpha}=g_-(a^\alpha,b^\alpha)\,,
\end{equation}
and hence the map that we constructed by  \eqref{Eq:hatinv}  is indeed the local inverse of $F$.
\end{proof}
It should be noted that although the map $F$ from Theorem \ref{thm:X1} is only a local diffeomorphism, the formulae \eqref{hatPB} yield
a holomorphic Poisson structure on the full space $S(n,d)$.

\begin{rem}
By using $\cG_\pm$ \eqref{cGpmdef}, the formula
\be
(a,b) \mapsto (\cG_+(a,b), \cG_-(a,b))
\ee
defines a local  Poisson map from $(S(n,1), \{\ ,\ \}_\kappa)^{\times d}$ to $( \Gn^*, \{\ ,\ \}_{*}^\kappa)$.
This map  satisfies  the identity
\be
\cG_+(a,b)^{-1} \cG_-(a,b) = \1_n - \hA(a,b) \hB(a,b),
\ee
which is a counterpart of the identity \eqref{factid2}.
The left-action of $\Gn$ on $S(n,d)$ has the Poisson--Lie property with
respect to the bracket $\{\ ,\ \}_\kappa'$ \eqref{hatPB} on $S(n,d)$ and the bracket
$\{\ ,\ \}_G^{-\kappa}$ \eqref{I6} on $\Gn$. The map $(\hA,\hB) \mapsto \hat \Gamma(\hA,\hB):= \1_n - \hA \hB$
represents the (densely defined) moment map associated with this action, using the standard mapping \eqref{chi}
of $\Gn^*$ into $\Gn$.
To put it more explicitly, $\hA, \hB$ and $\hat \Gamma$ satisfy
\be
\{ \hat \Gamma_1, \hat \Gamma_2\}'_\kappa =  -\kappa \left(  \hat \Gamma_1 r_-^n \hat \Gamma_2 + \hat \Gamma_2 r_+^n \hat \Gamma_1 - \hat \Gamma_1 \hat \Gamma_2 r^n - r^n \hat \Gamma_1 \hat \Gamma_2 \right)
\label{Ga1prime}\ee
together with
\be
\{ \hA_1, \hat \Gamma_2\}_\kappa' = -\kappa (\hat \Gamma_2 r_+^n - r_-^n \hat \Gamma_2) \hA_1
\quad\hbox{and}\quad
\{ \hB_1, \hat \Gamma_2\}_\kappa' = -\kappa \hB_1 (r_-^n \hat \Gamma_2 - \hat \Gamma_2 r_+^n).
\label{Ga2prime}\ee
These relations have the same form as those in \eqref{Ga1} and \eqref{Ga2}, taking into account that now
we are referring to Poisson--Lie symmetry with respect to the bracket $\{\ ,\ \}_G^{-\kappa}$ on $\Gn$.
The interested reader can verify all these equalities by direct calculation.
\end{rem}

Finally, we state the sought after decoupling property of the Arutyunov--Olivucci Poisson bracket \eqref{+PB}.

\begin{cor} \label{cor:X3}
Let $\theta_A$ and $\theta_B$ be arbitrary constants satisfying $\theta_A \theta_B = - \frac{1}{\kappa}$.
Then the  rescaling
\be
\theta: (\hA, \hB) \mapsto (\cA,\cB):= (\theta_A \hA, \theta_B \hB)
\ee
gives a Poisson diffeomorphism from $(S(n,d), \{\ ,\ \}_\kappa')$ \eqref{hatPB} to $(S(n,d), \{\ ,\ \}_\kappa^+)$ \eqref{+PB}.
Composing this with the map $F$ from Theorem \ref{thm:X1}, we get a local Poisson diffeomorphism
\be
\theta \circ F: (S(n,1), \{\ ,\ \}_\kappa)^{\times d} \to (S(n,d), \{\ ,\ \}_{\kappa}^+).
\ee
This map enjoys the identity
\be
(\1_n + \kappa \cA \cB) \circ \theta \circ F   = \cG_+^{-1} \cG_-.
\label{AOres}\ee
\end{cor}

The observation that $(\1_n + \kappa \cA \cB)$ can be  realized by applying the mapping \eqref{chi} on the inverse $(\cG_+, \cG_-)^{-1}$ of a Poisson map $(\cG_+, \cG_-)$ into
$(\Gn^*, \{\ ,\ \}_{*}^\kappa)$ played an important role in the derivation of the trigonometric
complex spin Ruijsenaars--Schneider model by Arutyunov and Olivucci \cite{AO}.
(To be precise, they locally realized $(\cG_+, \cG_-)$ as a  moment map  generating a Poisson--Lie
 action of $(\Gn, \{\ ,\ \}_G^\kappa)$ on $(S(n,d), \{\ ,\ \}_\kappa^+)$.)
Our result \eqref{AOres} provides decoupled variables (the $(a^\alpha, b^\alpha)$ for $\alpha=1,\dots, d$) that give such a realization explicitly.
These new variables $(a,b)$ are expected to be useful for further studies of the reduction
treatment of the complex spin Ruijsenaars--Schneider model, similarly
 as proved to be the case
for the real form of this important integrable Hamiltonian system \cite{FFM}.

\section{Conclusion} \label{Sec:Conc}

In this paper we presented a detailed analysis of the $\Gn \times \Gd$
covariant Poisson structures \eqref{I4} and \eqref{+PB} on the linear space $S(n,d)$ \eqref{I1} for arbitrary natural numbers $n$ and $d$.
Our main results are encapsulated by Theorem \ref{Thm:Main}  and Theorem \ref{thm:X1} with Corollary \ref{cor:X3} that provide new realizations
of the corresponding Poisson algebras in terms of $d$ independent copies of `elementary spin variables' living in $S(n,1)$.
The subsequent appendix highlights further relevant properties of these Poisson structures, especially by giving
the underlying symplectic form on a dense open subset of $S(n,1)$.
These results may contribute, for example, to deepening
the understanding of  Ruijsenaars--Schneider type integrable many-body models with spin having hidden $\Gn$ Poisson--Lie  symmetry.
It is also an interesting open question to search for their quantum mechanical analogues in the future.

\bigskip
\noindent
{\bf Acknowledgements.}
The research of M.F. was supported by a Rankin-Sneddon Research Fellowship of the University of Glasgow.
The work of L.F. was supported in part by the NKFIH research grant K134946.



\appendix

\section{Additional properties of the Poisson bracket on \texorpdfstring{$S(n,d)$}{S(n,d)}} \label{Sec:App}

In this appendix, we show that the Poisson bracket \eqref{I4} on $S(n,1)$ can be seen as a particular example of
 the complexification  of Zakrzewski's  $\U(n)$ covariant Poisson brackets on $\CC^n$ \cite{Z96}.
We also point out that the Poisson bracket \eqref{I4} on $S(n,d)$ is never globally symplectic,
and present the symplectic form that corresponds to this Poisson bracket in a neighborhood of zero.

\subsection{The Zakrzewski Poisson brackets}

Let us introduce a real anti-symmetric biderivation on $\CC^n\simeq \R^{2n}$, which is written on the   components of
 $u\in \CC^n$ as
\begin{equation}
\begin{aligned}
  \br{u_i,u_j}=&-\epsilon \ic \sgn(i-j) u_i u_j \,, \\
  \br{u_i, \bar u_l}=&-\epsilon  \ic  \delta_{il} F+ \epsilon\ic G u_i \bar u_l-\epsilon \ic \delta_{il} \sum_{r=1}^n \sgn(r-i) |u_r|^2\,, \label{Eq:ZakUn}
\end{aligned}
  \end{equation}
 where $\epsilon \in \R^*$, and $F=F(|u|^2)$, $G=G(|u|^2)$ are two arbitrary functions. Denote by $F',G'$ the derivatives $F'(t)=\frac{d}{dt}F(t)$, $G'(t)=\frac{d}{dt}G(t)$.  We recall the following result due to Zakrzewski.
\begin{lem} \label{Lem:ZakReal}
\emph{(\cite{Z96})} The anti-symmetric biderivation \eqref{Eq:ZakUn} is always Poisson for $n=1$, while for $n\geq 2$ it is Poisson
if and only if
\begin{equation} \label{Eq:ZakCond}
 FF'+G(F-F' |u|^2)=|u|^2\,.
\end{equation}
Furthermore,
 the action $\U(n)\times \CC^n \to \CC^n$ defined
by left multiplication yields a Poisson map if $\U(n)$ is equipped with
the multiplicative Poisson bracket satisfying $\{g_1, g_2\}_{\U(n)} = -2\ic \epsilon [g_1 g_2, r^n]$.
\end{lem}

In complete analogy, we define an anti-symmetric, holomorphic   biderivation on $\CC^{2n}$ endowed with coordinates $(a_i,b_i)$ (where we see $a$ and $b$ respectively as a vector and a covector) by
\begin{subequations}
 \begin{align}
  \br{a_i,a_j}=&\frac{\kappa}{2} \sgn(i-j) a_i a_j \,,\quad
  \br{b_i,b_j}=-\frac{\kappa}{2} \sgn(i-j) b_i b_j\,, \label{Eq:PBZakaa}\\
  \br{a_i, b_l}=&\frac{\kappa}{2}   \delta_{il} F- \frac{\kappa}{2} G a_i b_l+\frac{\kappa}{2}  \delta_{il} \sum_{r=1}^n \sgn(r-i) a_r b_r\,,\label{Eq:PBZakbb}
 \end{align}
\end{subequations}
 where $\kappa \in \CC^*$, and $F=F(t)$, $G=G(t)$ are two arbitrary holomorphic functions of $t:=\sum_{r=1}^n a_rb_r$. We denote by $F',G'$ the derivatives of $F,G$ with respect to $t$. The following result can then be proved  as Lemma \ref{Lem:ZakReal}.
\begin{lem} \label{Lem:ZakComp}
The anti-symmetric biderivation \eqref{Eq:PBZakaa}--\eqref{Eq:PBZakbb} is always Poisson for $n=1$, while for $n\geq 2$ it is Poisson
if and only if
\begin{equation} \label{Eq:ZakCondComp}
 FF'+G(F-F' t)=t\,, \qquad t:=\sum_{r=1}^n a_rb_r\,.
\end{equation}
Furthermore, the action
\begin{equation}
\tau: \Gn \times \CC^{2n} \to \CC^{2n}\,,\quad \tau_g(a,b)=(ga,b g^{-1})\,,
\end{equation}
is a Poisson map
when $\Gn$ is equipped with the Poisson bracket \eqref{I6}.
\end{lem}

The Poisson bracket \eqref{I4} on $S(n,1)$ is an example of  the complex Zakrzewski Poisson brackets of Lemma \ref{Lem:ZakComp},
since it corresponds to the case $F(t)=2+t$ and $G(t)=-1$.

\begin{rem}  \label{Rem:inv}
The involution
 \begin{equation} \label{Eq:iota}
  \iota : \CC^{2n}\to \CC^{2n}\,, \quad
   \iota(a,b)=(b^T, a^T)\,,
 \end{equation}
is an anti-Poisson automorphism.
This follows from a direct verification of this property on the evaluation functions $a_i,b_i$, see \eqref{Eq:PBZakaa}--\eqref{Eq:PBZakbb}.
\end{rem}

\subsection{Degeneracy of the Poisson bracket}

We note that the Poisson bracket \eqref{I4} on $S(n,d)$ can be written in the coordinates $A_i^\alpha:=A_{i\alpha}$, $B_i^\alpha:=B_{\alpha i}$, as
 \begin{subequations}
  \begin{align}
 \br{A_i^\alpha,A_k^\beta}=&\frac{\kappa}{2}\sgn(i-k) A_k^\alpha A_i^\beta -\frac{\kappa}{2} \sgn(\alpha-\beta) A_k^\alpha A_i^\beta \,, \label{Eq:pbAA}\\
  \br{B_i^\alpha,B_k^\beta}=&-\frac{\kappa}{2}\sgn(i-k) B_k^\alpha B_i^\beta +\frac{\kappa}{2} \sgn(\alpha-\beta) B_k^\alpha B_i^\beta \,, \label{Eq:pbBB}\\
   \br{A_i^\alpha,B_k^\beta}=&\frac{\kappa}{2} \delta_{ik} A_i^\alpha B_k^\beta + \kappa \delta_{ik} \sum_{s>i} A_s^\alpha B_s^\beta \nonumber \\
&+\frac{\kappa}{2} \delta_{\alpha\beta} A_i^\alpha B_k^\beta + \kappa\delta_{\alpha\beta} \sum_{\mu<\alpha} A_i^\mu B_k^\mu + \kappa\delta_{\alpha\beta} \delta_{ik}      \,. \label{Eq:pbAB}
  \end{align}
 \end{subequations}
If we consider a point $p\in S(n,d)$ contained in the one-dimensional subset where
\begin{equation}
 1+A^1_n B^1_n=0\,, \quad A^\alpha_j,B^\alpha_j=0\, \,\, \text{for }(j,\alpha)\neq (n,1)\,,
\end{equation}
we directly see that the Poisson brackets \eqref{Eq:pbAA}--\eqref{Eq:pbBB} evaluated at  $p$ are zero. Furthermore, decomposing \eqref{Eq:pbAB} as
\begin{subequations}
  \begin{align}
\br{A_i^\alpha,B_k^\beta}=& 0\,, \qquad i\neq k,\alpha \neq \beta\,, \label{Eq:ApE0} \\
\br{A_i^\alpha,B_k^\alpha} =& \frac{\kappa}{2} A_i^\alpha B_k^\alpha + \kappa \sum_{\mu<\alpha} A_i^\mu B_k^\mu \,, \quad  i\neq k,\, \alpha=\beta\,, \label{Eq:ApE1} \\
\br{A_i^\alpha,B_i^\beta}=& \frac{\kappa}{2}  A_i^\alpha B_k^\beta + \kappa  \sum_{s>i} A_s^\alpha B_s^\beta \,, \quad i=k,\, \alpha\neq \beta\,, \label{Eq:ApE2} \\
   \br{A_i^\alpha,B_i^\alpha}=& \kappa A_i^\alpha B_i^\alpha + \kappa  \sum_{s>i} A_s^\alpha B_s^\alpha
+ \kappa  \sum_{\mu<\alpha} A_i^\mu B_i^\mu + \kappa  \,, \quad i=k,\,\alpha=\beta\,, \label{Eq:ApE3}
  \end{align}
\end{subequations}
we get that \eqref{Eq:ApE0}--\eqref{Eq:ApE2} vanish at $p$, while  we can write \eqref{Eq:ApE3} as
\begin{equation}
 \br{A_i^\alpha,B_i^\alpha}(p)=\kappa [1-\delta_{in}]\,[1-\delta_{\alpha 1}]\,.
\end{equation}
Hence, at $p$ the rank of the Poisson structure is $2(n-1)(d-1)$ and the Poisson bracket \eqref{I4} is not globally non-degenerate.

\subsection{The symplectic form on a dense open subset}

The  holomorphic Poisson bracket \eqref{I4} is non-degenerate on a dense  subset of $S(n,d)$, since it is non-degenerate at the origin.
 We now present the corresponding symplectic form for $d=1$. For $d\geq 2$, the symplectic form can be obtained around the origin by combining this result with Theorem \ref{Thm:Main}.
\begin{prop} \label{pr:Symp}
Consider  $S(n,1)$ with the Poisson bracket \eqref{I4}, and denote its elements by $(a,b)$.
 On the open subset where
\begin{equation}
G_i:=1+\sum_{r\geq i}a_rb_r \neq 0\,, \quad \forall\, 1 \leq i \leq n\,,
\end{equation}
the Poisson structure is non-degenerate, and the associated symplectic form can be  written as
\begin{equation}
 \begin{aligned} \label{Eq:SympF}
  \omega =-\frac{1}{\kappa} \sum_{i=1}^n \frac{da_i \wedge db_i}{G_i}
  +\frac{1}{2\kappa} \sum_{i=1}^n \sum_{s>i} \frac{1}{G_i G_{i+1}}
  \left(b_i da_i- a_i db_i \right) \wedge \left(b_s da_s+a_s db_s \right)\,,
 \end{aligned}
\end{equation}
where we set $a_{n+1}=b_{n+1}=0$ and $G_{n+1}=1$.
\end{prop}
\begin{proof}
Without loss of generality, we take $\kappa=2$.
We will prove using the $2$-form \eqref{Eq:SympF} that for any $1\leq j \leq n$
\begin{equation} \label{Eq:contract}
 \iota_{X_{a_j}}\omega=-da_j\,,
\end{equation}
where $X_{a_j}=\br{-,a_j}$ denotes the Hamiltonian vector field of $a_j$, which is given by
\begin{equation} \label{Eq:Xaj}
 X_{a_j}=\sum_{l\neq j} \sgn(l-j) a_ja_l\frac{\partial}{\partial a_l}
 -\sum_{l\neq j} a_j b_l\frac{\partial}{\partial b_l} - 2 G_j \frac{\partial}{\partial b_j}\,.
\end{equation}
By symmetry between $a$ and $b$, we will also have that $\iota_{X_{b_j}}\omega=-db_j$. These conditions then imply that $\omega$ is non-degenerate and  corresponds to the Poisson bracket on $S(n,1)$, hence it is also closed.

To prove \eqref{Eq:contract}, let us denote the three terms appearing in the vector field \eqref{Eq:Xaj} as $X_1,X_2,X_3$. Contracting with the $2$-form, we compute
\begin{equation}
 \begin{aligned} \label{Eq:EX1}
\iota_{X_1}\omega=& -\frac12 \sum_{l\neq j} \sgn(l-j)  \frac{a_j a_l}{G_l} db_l \\
&+\frac14 \sum_{l\neq j} \sgn(l-j) \sum_{s>l} a_j\frac{a_l b_l}{G_l G_{l+1}} \left(b_s da_s+a_s db_s \right) \\
&-\frac14 \sum_{l\neq j} \sgn(l-j) \sum_{i<l} a_j\frac{a_l b_l}{G_i G_{i+1}} \left(b_i da_i- a_i db_i \right)\,,
 \end{aligned}
\end{equation}
then
\begin{equation}
 \begin{aligned} \label{Eq:EX2}
\iota_{X_2}\omega=& -\frac12 \sum_{l\neq j}  \frac{a_j b_l}{G_l} da_l
+\frac14 \sum_{l\neq j}  \sum_{s>l} a_j\frac{a_l b_l}{G_l G_{l+1}}\left(b_s da_s+a_s db_s \right)  \\
&+\frac14 \sum_{l\neq j}  \sum_{i<l} a_j\frac{a_l b_l}{G_i G_{i+1}} \left(b_i da_i- a_i db_i \right)\,,
 \end{aligned}
\end{equation}
and finally
\begin{equation}
 \begin{aligned}
\iota_{X_3}\omega=& -da_j
+\frac12  \sum_{s>j} a_j\frac{1}{G_{j+1}} \left(b_s da_s+a_s db_s \right)  \\
&+\frac12 \sum_{i<j} a_j\frac{G_j}{G_i G_{i+1}} \left(b_i da_i- a_i db_i \right)\,.
 \end{aligned}
\end{equation}
After summing together the last two terms from \eqref{Eq:EX1} and \eqref{Eq:EX2}, we find
\begin{equation}
 \begin{aligned} \label{Eq:EXaj}
\iota_{X_{a_j}}\omega=&-da_j -\frac12 \sum_{s\neq j} \sgn(s-j)  \frac{a_j a_s}{G_s} db_s -\frac12 \sum_{s\neq j}  \frac{a_j b_s}{G_s} da_s \\
&+\frac12  \sum_{s>l>j} a_j\frac{a_l b_l}{G_l G_{l+1}} \left(b_s da_s+a_s db_s \right)
+\frac12 \sum_{s<l<j} a_j\frac{a_l b_l}{G_s G_{s+1}} \left(b_s da_s-a_s db_s \right) \\
&+\frac12  \sum_{s>j} a_j\frac{1}{G_{j+1}} \left(b_s da_s+a_s db_s \right)
+\frac12 \sum_{s<j} a_j\frac{G_j}{G_s G_{s+1}} \left(b_s da_s-a_s db_s \right)\,.
 \end{aligned}
\end{equation}
For fixed $s,j$, we note the  identities
\begin{equation}
 \sum_{l=j+1}^{s-1}\frac{a_l b_l}{G_l G_{l+1}} =\frac{1}{G_s}-\frac{1}{G_{j+1}}\,, \qquad
 \sum_{l=s+1}^{j-1}a_lb_l = G_{s+1}-G_j\,,
\end{equation}
which allow us to write the line in the middle of \eqref{Eq:EXaj} as
\begin{equation}
+\frac12  \sum_{s>j} a_j \left(\frac{1}{G_s}-\frac{1}{G_{j+1}} \right) \left(b_s da_s+a_s db_s \right)
+\frac12 \sum_{s<j} a_j\frac{G_{s+1}-G_j}{G_s G_{s+1}}  \left(b_s da_s-a_s db_s \right)\,.
\end{equation}
This can be simplified with the last line of  \eqref{Eq:EXaj}, and we get
\begin{equation}
 \begin{aligned} \label{Eq:EXajB}
\iota_{X_{a_j}}\omega=&-da_j -\frac12 \sum_{s\neq j} \sgn(s-j)  \frac{a_j a_s}{G_s} db_s -\frac12 \sum_{s\neq j}  \frac{a_j b_s}{G_s} da_s \\
& +\frac12  \sum_{s>j} a_j  \frac{a_sb_s}{G_s} \left(\frac{da_s}{a_s}+\frac{db_s}{b_s} \right)
+\frac12 \sum_{s<j} a_j\frac{ a_s b_s}{G_s} \left(\frac{da_s}{a_s}-\frac{db_s}{b_s} \right) \,,
 \end{aligned}
\end{equation}
which is just $-da_j$ as the other four terms cancel out.
\end{proof}

\begin{rem}
It can be shown that the Poisson tensor corresponding to
the bracket \eqref{I4} on $S(n,1)$
is degenerate precisely on the zero set of the function $\prod_{i=1}^n G_i$,
which is the complement of the set considered in Proposition \ref{pr:Symp}.
If in  \eqref{Eq:SympF} we put $b_j = \bar a_j$ and $\kappa = 2\ic$, then we recover the real symplectic form on $\C^n \simeq \R^{2n}$
given by \cite[Proposition A.6]{FFM}, from which our formula was obtained by complexification.
\end{rem}


\end{document}